\documentclass[11pt,reqno]{amsart}
\usepackage{amsmath,amsthm}
\usepackage{latexsym}
\usepackage{amsfonts}
\usepackage{amssymb}
\usepackage[dvipsnames]{xcolor}
\usepackage{graphicx}
\usepackage{float}
\usepackage{graphbox}
\usepackage{textcomp}
\usepackage[pagebackref, colorlinks = true, linkcolor = blue, urlcolor  = blue, citecolor = red]{hyperref}
\usepackage[margin=1in]{geometry}
\usepackage{array}
\usepackage{comment}

\renewcommand{\epsilon}{\varepsilon}
\renewcommand{\phi}{\varphi}

\newtheorem{theorem}{Theorem}[section]
\newtheorem{proposition}[theorem]{Proposition}
\newtheorem{lemma}[theorem]{Lemma}

\newtheorem{remark}[theorem]{Remark}

\newcommand{\ket}[1]{|#1\rangle} %ket
\newcommand{\bra}[1]{\langle#1|} %bra
\newcommand{\circM}{\operatorname{circ}}

\definecolor{candypink}{rgb}{0.89, 0.44, 0.48}

\begin{document}

\title[  An operator-algebraic approach for generalized Cardano polynomials]{  An operator algebraic approach for generalized Cardano polynomials}

\author{Leonard Mada}
\address{Department of Immunology, Victor Babe\c{s} University of Medicine, Eftimie Murgu Square 2, 300041, Timi\c{s}oara, Romania.}
\email{lmada@umft.ro}

\author{Maria Anastasia Jivulescu}
\address{Department of Mathematics, Politehnica University of Timi\c soara, Victoriei Square 2, 300006 Timi\c soara, Romania.}
\email{maria.jivulescu@upt.ro}

%\author{}
%\address{}
%\email{}

\begin{abstract}

We develop an operator‑algebraic framework for generalized Cardano polynomials and show how their structure naturally leads to an operator formulation of Cardano’s method that is compatible with tools and concepts from quantum information theory (QIT). The generalized Cardano polynomials are constructed as a generalization of classical theory of Cardano formula for cubic equation, as well as through the spectral properties of the circular operator $W=pZ_n+qZ_n^{-1},$ where $Z_n$ is the n-dimensional clock operator. This representation embeds Cardano‑type identities into the spectral theory of circulant operators.
The construction clarifies the algebraic structure and solvability of a family of two-parameters  odd-order polynomials,   classically and through operator methods familiar in QIT.  As application, spectral calculus on operator algebras is used to construct Cardano-generated POVMs.
\end{abstract}
%\date{\today}

\maketitle 

%\tableofcontents 
% MAIN BODY: Replace "Section Title," "Subsection Title," and "Subsubsection Title" with your section, subsection, and subsubsection titles, using title case to capitalize the first and all major words in these titles. Replace lorem ipsum text with your text for each section and subsection. Add additional sections and subsections as needed.
\section{Introduction} The study of algebraic equations has a long history and classical formulas for solving cubic  equations are well-known due to 
 Cardano\cite{Katz2009}. The method  consists in reducing the original polynomial $x^3+ax^2+bx+c$ into a simpler form, called depressed, by replacing $x\mapsto x-(a/3)$. %Consequently, it follows $
  %  x^3+\alpha x+\beta=0,\, \alpha=b-\frac{a^2}{3}$ and $\beta=c-\frac{ab}{3}+\frac{2a^3}{27}$.
 \\The famous \textbf{Cardano formula} refers to the solutions of the equation \begin{equation}\label{Cardano}
     x^3+3sx+t=0, \text{ where }  s=\frac{b}{3}-\frac{a^2}{9} \text{ and }t=c-\frac{ab}{3}+\frac{2a^3}{27}, 
 \end{equation}  
 that are given as
 \begin{equation}
     x_1=u_0+v_0,\, x_1=\omega u_0+\omega^2v_0 \text{ and }
     x_2=\omega^2u_0+\omega v_0,
 \end{equation}
 where $\omega=e^{\frac{2\pi i}{3}}$ and $u_0,$ $v_0$ are primitive solutions of equations
\begin{equation}
    u^3=\frac{-t+\sqrt{t^2+4s^3}}{2} \text{ and } v^3=\frac{-t-\sqrt{t^2+4s^3}}{2}, \text{ where } \Delta:= t^2+4s^3\geq 0.
\end{equation}

Over the past several years, Cardano formulas have been the subject of further studies \cite{Pfeifer2024}. We  mention here simplifications of  Cardano's solutions \cite{Liao2023}, adaptations of Cardano's approach to Vieta-Lucas polynomials and Vieta-Finonacci functions\cite{Witula2010}, as well as more modern approaches, such as the use of superposition principle of quantum mechanics \cite{Fujii2013}. In
Fujii's reinterpretations of Cardano's formula, an operator  formalism has been developed, where spectral properties of low-dimensional operators are used to reproduce the classical solutions. In addition, generalization to Cardano polynomials and its applications to radical reductions are known due to \cite{Osler2001}.
\par In this paper we generalize Cardano's approach  to recover an odd order $(n=2m+1)$ generalized  polynomial and its roots.  Also,  Fujii's operator techniques for the cubic and quartic case are extended to a family of two parameters odd degree polynomials, based on the observation that, given two real numbers $c,d$, then two variables $p$ and $q,$ that satisfy the relations $pq=c,\,\, p^n+q^n=2d$, form the root  $x=p+q$ of a generalized Cardano polynomial. The Fujii operator $$W=pZ_n+qZ_n^{-1},$$
where $Z_n$ is the n-dimensional clock operator, satisfies the general odd-degree Cardano-operator polynomial
$$W^n=2dI+\sum\limits_{i=0}^{m-1}B_{m,j}c^{m-j}W^{2j+1},$$
where the coefficients $B_{m,j}$ are derived explicitly. Our operator formalism is based on the observation that the conjugation of the Fujii operator $W$ with the discrete Fourier transform $F_n$, 
transforms it into  $X=F_n^+WF_n, $ a  circulant operator of the form $X=pX_n+qX_n^{-1},$ where $X_n$ is the cycle shift operator. Circulant operators play a central role in quantum Fourier analysis, quantum walks and Hamiltonian simulations and their spectral properties are completely determined by the Fourier transform. Indeed,  the eigenvalues of $X$ are $$\lambda_k=p\omega^k+q\omega^{-k},\, \omega=e^{2\pi i/n},$$
which are exactly the roots of generalized Cardano polynomials. Therefore, we extend the operator Cardano framework to two-parameter family of odd degree operator polynomials, using an algebra generated by clock, shift and Fourier operators \cite{Gosson2006}. Our approach provides a new connection as well as a quantum formulation of the classical algebraic theory of two-parameters family of odd polynomials, using structures specific to Quantum Information Theory \cite{NC2010}, such as the operator algebra generated by clock and shift operators. 
Moreover, we generate  positive-operator valued measurements (POVMs), using spectral decomposition of the Cardano operator associated to a given polynomial.

 The paper is organized as follows. In Section \eqref{S:gen}  we introduce algebraically the odd-degree Cardano  polynomials, presenting the general theory, examples, as well as its connections to Chebyshev polynomials. In Section \eqref{S:Op} we consider an operator version of Cardano polynomials and show how this structure encodes the solutions in the operators $W$  as the eigenvalues of the Cardano operator. Section \eqref{sec:CardanoPOVMs} is devoted to the construction of Cardano-generated POVMs and explicit examples are presented.
 \vspace{0.3cm}
\section{Generalised Cardano  polynomials}\label{S:gen} 
In the following section we aim to derive the general solution to a class of odd-degree polynomials, which are a natural generalization of the well-known three-order  Cardano polynomial. To this aim we present here few mathematical results.
\subsection{Algebraic approach to Cardano polynomials.}
\begin{lemma}
    Given $c $ and  $d, $ two real parameters( $D:=d^2-c^n\geq 0$) and $n,$ an odd natural number,   the system of equations \begin{equation}\label{prop-p-q}
    \begin{cases}
        x^n+y^n=2d\\
         xy=c
    \end{cases}
     \end{equation}
has the set of solutions $\{(p\omega^{j},\,q\omega^{-j})\},\, j\in\{0,\pm 1,\ldots \pm [(n-1)/2]\},$  where
\begin{equation}\label{defp-q}
    p:=\sqrt[n]{d+\sqrt{d^2-c^n}} ,\,\, q:=\sqrt[n]{d-\sqrt{d^2-c^n}},\end{equation}
and     $\omega=e^{\frac{2\pi i}{n}}$ is the n-th order root of unity, ($\omega^n=1$).
\end{lemma}
  \begin{proof}
We notice that, given that the discriminant $D$, defined as $D:=d^2-c^n\geq 0$, the parameters $p$ and $q$ are well defined (the case $D<0$ will be treated separately, in Subsection{\eqref{s: neg-D}}).
Also, the system \eqref{prop-p-q}, in the unknowns $x$ and $y$, is a  symmetric system. Straightforward computations show that the pairs $(p\omega^{j},\,q\omega^{-j}), j\in\{0,\pm 1,\ldots \pm [(n-1)/2]\},$ satisfy the equations \eqref{prop-p-q}. In this way we recover  one real solution and $[(n-1)/2]$ pairs of complex conjugate solutions.
  \end{proof}

Let us denote \begin{equation}\label{j-root} x[j]:=p\omega^{j}+q\omega^{-j},j\in\{0,\pm 1,\ldots \pm [(n-1)/2]\}\}.\end{equation}  We aim to recover the n-th order equation of real coefficients (for odd n), whose solutions are given by $x[j],\,j\in \{0,\pm 1,\ldots \pm[(n-1)/2]\}$.
To this aim, we introduce the procedure for simple cases, such as $n=3$ and $n=5$.
\subsection{3-rd and 5-th order Cardano polynomials}
For $n=3$, it is easy to show that $x[0]=p+q$ is a solution of the depressed form of the cubic equation \begin{equation}
    \label{3-eq}x^3-3cx-2d=0.
\end{equation} Indeed,  for $x=x[0]=p+q$ and using straightforward algebraic computation, it follows that
\begin{eqnarray}
    x^3=(p+q)^3 = p^3 + q^3 + 3pq(p+q)\Rightarrow x^3=3cx+2d.
\end{eqnarray}
Therefore,  $x[0]=p+q$ is a solution of the cubic equation $x^3-3cx-2d=0$. 
In a similar way, it can be shown that $x[1]=p\omega+q\omega^{-1}$ and $x[-1]=p\omega^{-1}+q\omega$ are both solutions of eq. \eqref{3-eq}.
Hence, for given $c$ and $d$, the cubic polynomial of Cardano  is $$C_{3,c,d}(x)=x^3-3cx-2d$$ and its roots are $\{p+q,p\omega+q\omega^{2},p\omega^{2}+q\omega \}$.
\\In a similar way we find the 5-th order polynomial equation; indeed, by setting $x=x[0]=p+q,$
\begin{eqnarray}\nonumber 
    x^5=(p+q)^5 = p^5+q^5+5pq(p^3+q^3)+10p^2q^2(p+q)
\end{eqnarray}
Using the relations $p^5+q^5=2d, pq=c$ and $p^3+q^3=x^3-3cx$, it follows that
$x^5=5cx^3-5c^2x+2d.$ 
Again,  for given $c$ and $d$,  the fifth polynomial of Cardano  is $$C_{5,c,d}(x)=x^5-5cx^3+5c^2x-2d$$ and its roots are $\{p+q,p\omega+q\omega^{4},p\omega^{4}+q\omega, p\omega^2+q\omega^{3}, p\omega^3+q\omega^{2} \}$.
\\We understand that the previous procedure can be extended to any arbitrary odd natural number $n$; that is, based on the real numbers $c$ and $d$  and the associated parameters $p$ and $q$, as given by \eqref{defp-q}, and their properties \eqref{prop-p-q},  we can find the odd order polynomial equations satisfied by $x=p+q$, using a  simplified procedure, similar to the Cardano formula. More exactly, based on the binomial formula of n-th order for $(p+q)$, we can express  all the terms of decompositions in terms of the sum $p+q$, as well as of product $pq=c$ and, in this way, we will recover the n-th order Cardano  polynoimial. Let's denote it by $C_{n,c,d}(x)$.   We present here in brief the first few polynomials of Cardano  (for given parameters $c$ and $d$):

\begin{eqnarray}
    C_{3,c,d}(x)=x^3 - 3cx - 2d  &\\
    C_{5,c,d}(x)=x^5 - 5cx^3 + 5c^2x - 2d  &\\
    C_{7,c,d}(x)=x^7 - 7cx^5 + 14c^2x^3 - 7c^3x - 2d .
\end{eqnarray} 

%A similar approach for n-even, will provide the following series of  polynomial equations:
%\begin{eqnarray}
 %   x^2 - 2c - 2d = 0 &\\
  %  x^4 - 4cx^2 + 2c^2 - 2d = 0 &\\
   % x^6 - 6cx^4 + 9c^2x^2 - 2c^3 - 2d = 0
%\end{eqnarray}
\subsection{The n-th order Cardano  polynomial}\label{s:2}
We aim to find the polynomial equations of odd order ($n=2m+1$), using a similar procedure, as shown above for $n=3$ and $n=5$. To this aim, we are going to state the following 
\begin{theorem} Consider $n,$ an odd natural number ($n=2m+1 $), and the real parameters $p$ and $q,$ as given by \eqref{defp-q}.
 The n-th order polynomial equation whose solutions are $x[j]=p\omega^j+q\omega^{-j}, j\in\{0,\pm1,\ldots, \pm [(n-1)/2]\}$ is
 \begin{eqnarray}
x^{2m+1}=2d+\sum\limits_{j=0}^{m-1}B_{m,j}c^{m-j}x^{2j+1}, \text{ where }\label{gen-eq}
\\ B_{m,j}=(-1)^{m-1-j}\frac{2m+1}{2j+1}\begin{pmatrix}
        m+j\\2j
    \end{pmatrix}, 0\leq j\leq m-1.
    \label{B}
\end{eqnarray}

      The polynomial \begin{equation}\label{Car-pol}
          C_{n,c,d}(x)=x^{2m+1}-\sum\limits_{j=0}^{m-1}B_{m,j}c^{m-j}x^{2j+1}-2d,
      \end{equation} associated to the poynomial equation \eqref{gen-eq}  will be called \textit{the n-th order generalized Cardano   polynomial}.
\end{theorem}
\begin{proof}
    Using the binomial formula for computing $x^n,$ where $x=p+q$ and $n=2m+1,$
    as well as Kummer's formula\cite{Ma2005} to express terms of the form $p^r+q^r,\, r\geq 1$, we express $x^n$ as in \eqref{gen-eq} in terms of coefficients \eqref{B}. For a complete proof, see Appendix \eqref{s:A}.    
    Hence, we obtain the n-th order polynomial equation as given by \eqref{gen-eq},  whose solution is $p+q=x[0]$. Moreover, also  $x[j]=p\omega^j+q\omega^{-j}, j\in\{0,\pm1,\ldots, \pm m)\}$ are solutions of \eqref{gen-eq} by the simple fact that $p\omega^{j}$ and $q\omega^{-j}$ satisfy the same properties \eqref{prop-p-q} as $p$ and $q$, with respect to their $(2m+1)$-th power sum and their product, expressions that appear also in eqs. \eqref{Krummer} and \eqref{bin-formula}. We conclude that, given equation \eqref{gen-eq}, we can compute their  solutions $x[\pm j],\,\, j\in\{0,\pm1,\ldots,\pm m\},$ as given by \eqref{j-root}. 
\end{proof}
\begin{remark}
    The values $x[j]$ will be referred as the j-th branches of the solution, as replacing $j$ by $j+k (\text{ mod } n)$ will only permute the set of roots.
\end{remark}

We present here a couple of examples of Cardano polynomials that can be recovered using \eqref{Car-pol}. 
Straightforward computations enable us to write the cubic Cardano polynomials, $C_{n,c,d}$, for given c, d, real parameters. Indeed, for $n=3$, we write $C_{3,c,d}(x)=x^3-B_{1,0}cx-2d,$ where $B_{1,0}=3 $  (as given by \eqref{B}, for $m=1,\, j=0$). So, $C_{3,c,d}(x)=x^3-3cx-2d$.
\\
 In a similar manner, for n=5, the polynomial $C_{5,c,d}(x)=x^5-\sum\limits_{j=0}^{1}B_{2,j}c^{2-j}x^{2j+1}-2d,$    where
$B_{2,0}=(-1)5\begin{pmatrix}
    2\\0
\end{pmatrix}=-5, $ and $B_{2,1}=\frac{5}{3}\begin{pmatrix}
    3\\2
\end{pmatrix}=5.$ Therefore, $C_{5,c,d}(x)=x^5-5cx^3+5c^2x-2d$. \\Similar computations allow us to recover high order Cardano polynomials, as given by \eqref{Car-pol}.
Moreover, our procedure allows to recover also the roots of the Cardano polynomials in a simpler manner, as given by \eqref{j-root}.
To this aim, we illustrate a few simple examples, focusing both on the construction of Cardano polynomial of a given order, starting from assessing explicit values of the parameters $c$ and $d$, as well as on  the method to find the complete set of solutions, for a given polynomial equation which can be cast in Cardano form. 
\\
\textit{Example 1.} Consider $c=d=1$. The values $p=q=1$ are obtained by straightforward computations, which also provide the equation for the associated cubic polynomial $x^3-3x-2=0$. Its solutions are $x[0]=2,\, x[1]=\omega+\omega^{-1}=-1$ and $x[-1]=\omega^{-1}+\omega=-1$.
\\For the same values of the parameters $c, d, p$ and $q$,  the fifth order polynomial equation is $x^5-5x^3+5x-2=0$ and the set of solutions is $x[0]=2, x[\pm 1]=\omega+\omega^4$ and $x[\pm 2]=\omega^2+\omega^3$.
\\ \textit{Example 2.}
   Let us consider the cubic polynomial equation   $x^3-6x-9=0$. The identification with the form \eqref{3-eq} gives us the following values for the parameters $c=2,\, d=9/2$; therefore, the parameters p and q, given by eqs. \eqref{defp-q}, are $p=2,\,q=1$. The solutions of the equation are given by \eqref{j-root}, as  
$x[j]=2\omega^j+\omega^{-j},\, j\in\{0,\pm 1\}.$ Straightforward computations provide the complete set of solutions\begin{center}
    $x[0]=3,\, x[1]=2\omega+\omega^2=\frac{-3+i\sqrt{3}}{2}$ and $x[-1]=2\omega^2+\omega=\frac{-3-i\sqrt{3}}{2}$.
   \end{center}

     \subsection{Negative discriminant case}\label{s: neg-D}
We  study in the following the case of negative  discriminant $D = d^2 - c^n< 0$. We  aim to find a closed form for the roots of the n-th order Cardano polynomial. In this case, the parameters p and q, as defined by eqs. \eqref{defp-q}, are complex conjugates
and are the solution of the system of equations given by \eqref{prop-p-q}.
Given that the parameter $d$ can be rewritten as  $d = c^{n/2}\cos(\alpha)$, where $\alpha = \arccos(\frac{d}{ c^{n/2}})$ (this is  possible,  given that $D<0$),  it holds that $$
D = d^2 - c^n = c^n\cos(\alpha)^2 - c^n =- c^n  sin^2(\alpha).$$
Therefore, \begin{equation}
    \sqrt{D} = i c^{n/2} \sin(\alpha).
\end{equation}
Applying de Moivre's formula, the parameters $p$ and $q$ are
\begin{eqnarray}
    p=\sqrt{c}(\cos \frac{\alpha}{n}+i\sin\frac{\alpha}{n}),\,\,
    q=\sqrt{c}(\cos \frac{\alpha}{n}-i\sin\frac{\alpha}{n})
     \label{p-q-neg}
\end{eqnarray}
Therefore, the j-th root of Cardano generalized polynomials, for the case $D<0$, given by the eq \eqref{j-root}, can be written as
\begin{eqnarray}\label{roots-d-neg}
    x[j]=p\omega^j+q\omega^{-j}=2\sqrt{c}\cos(\frac{\alpha+2\pi j}{n}), \,\forall j\in\{0,\pm 1\ldots,\pm [(n-1)/2]\}.
\end{eqnarray}
We notice that the case of n-th order Cardano polynomial, for which $D<0$, corresponds to real roots.
For $j=0$ we have that
\begin{equation}
x[0] = p + q = 2  c^{1/2}\cos[\alpha/n]
\end{equation}
\textit{Example} 
Consider the case $d = 2$ and $c = 3$ and $n = 5$. The associated 5-th order Cardano polynomial equation is given by
$x^5 - 15x^3 + 45x - 4=0$,
with $D = -239$. Its solutions are

$$x[0] = (2 + i\sqrt{239})^{1/5} + (2 - i\sqrt{239})^{1/5}\approx  3.321 $$
Using the trig-identities \eqref{roots-d-neg}, we recover the same set of solutions, which can be written in a more accessible form
\begin{equation}
x[j] = 2 \sqrt{3}  \cos[(\arccos(2/3^{5/2}) + 2\pi j)/5], \,\, j\in\{0,1,2\}.
\end{equation}
For $j=0,$ we have
\begin{eqnarray}
x[0] = 2 \sqrt{3} \cos[\arccos(2/3^{5/2})/5]\approx 3.321\nonumber
\end{eqnarray}

\subsection{Generalized Chebyshev Polynomials}\label{s:2.1}
We aim now to connect the theory of generalized Cardano polynomials to that of Chebyshev polynomials.
We recall that the so-called modified Chebyshev polynomials of the first kind are defined as:
\begin{equation}\label{def-Cb}
    \Omega_n(x):=\sum\limits_{k=0}^n(-1)^k \frac{n}{n-k}\begin{pmatrix}
        n-k\\k
    \end{pmatrix}x^{n-2k},\, n\in \mathbb{N}.
\end{equation}
The polynomial $\Omega_n(x)$, also called the n-th order Vieta-Lucas polynomial\cite{Witula2007}, satisfies the identity:
\begin{equation}
    \Omega_n(x)=2T_n(\frac{x}{2}),
\end{equation}
where $T_n(\cos \theta)=\cos (n \theta)$ is the n-th Chebyshev polynomial of first kind.
The modified Chebyshev polynomials of the first kind, $\Omega_n(x)$, have the following properties:
\begin{eqnarray}\label{recurence}
    \Omega_1(x)=x,\,  \Omega_2(x)=x^2-2 \nonumber \\
    \Omega_{n+2}(x) = x\Omega_{n+1}(x)-\Omega_n(x),  n\in\mathbb{N}. 
\end{eqnarray}
Using Theorem 4.1, from \cite{Witula2010}, we notice that Chebyshev polynomials of the first kind can be written as follows:
\begin{eqnarray}\label{CB1}
    (\sqrt{c})^n\Omega_n(\frac{x}{\sqrt{c}})-2d=\prod\limits_{k=0}
    ^{n-1}(x-p\omega^k-q\omega^{-k})
\end{eqnarray}
The right hand side term of \eqref{CB1} is, actually, the decomposition of the n-th order Cardano generalized polynomial \eqref{gen-eq}, $C_n[x]$, we establish the following representation of modified Chebyshev polynomials of the first kind $(n=2m+1)$:
\begin{eqnarray}
    \Omega_n(x)=c^{-n/2}C_{n,c,d}(\sqrt{c}x)+2dc^{-n/2}.
\end{eqnarray}
Using the recurrence relation \eqref{recurence}, we recover the following recurrence relation between generalized Cardano polynomials
\begin{equation}
    C_{{n+2},c,d}(x)-xC_{{n+1},c,d}(x)+cC_{n,c,d}(x)=2d(x-c-1), n\geq 1,
\end{equation}
which can be used to recover the generalized even-order Cardano polynomials, in terms of odd order Cardano polynomials, $C_{n,c,d}(x),$ as given by \eqref{Car-pol}.
\section{Cardano operator}\label{S:Op}

In the following we are going to generalize the results introduced in\cite{Fujii2013} that consist in finding the roots of cubic  Cardano polynomial using operator setting. 
\subsection{Fujii operator}
Firstly let's recall few important notions from Quantum Information Theory (QIT)\cite{NC2010}. In a given n-dimensional Hilbert space $\mathcal{H}=\mathbb{C}^n,$  the computational basis is given by $\{\ket{0},\ket{1},\ldots, \ket{n-1}\},$ where Dirac notation is used to denote the j-th vector $\ket{j}=(
    0\,  \ldots 1 \ldots 0)^T $  of the computational basis. The generalized Pauli operators are defined as:
\begin{eqnarray}
    \text{ Shift operator: } X_n\ket{j}=\ket{ j+1} (\text{ modulo } n),
    \\ \text{ Clock operator: } Z_n\ket{j}=\omega^j\ket{ j},
\end{eqnarray} 
where $\omega=e^{2\pi i/n}$ and $n\geq 2$. For example, for $n=2$ we recover the well-known Pauli operators $X_2=\sigma_x=\begin{pmatrix}
    0&1\\1&0
\end{pmatrix}$ and $Z_2=\sigma_z=\begin{pmatrix}
    1&0\\0&-1
\end{pmatrix},$ whereas for $n=3$ we have  $X_3=\begin{pmatrix}
    0&0&1\\1&0&0\\0&1&0
\end{pmatrix} $ and $Z_3=\text{ diag }[1, \omega,\omega^2]$.
%and its inverse is $X_3^{-1}=\begin{pmatrix}
   % 0&1&0\\0&0&1\\1&0&0
%\end{pmatrix}.$
Also, we have that $Z_nX_n=\omega X_nZ_n, \,\forall  n \geq 2$. 
\\The discrete Fourier transform is defined as
\begin{equation}
    F_n\ket{j}:=\frac{1}{\sqrt{n}}\sum\limits_{i=0}^{n-1}\omega^{ij}\ket{i},\, \forall  n \geq 2
\end{equation}
It holds that $F_n^+Z_nF_n=X_n$ and $F_n^+Z_n^{-1}F_n=X_n^{-1},$ where $(\cdot )^+$ is the conjugate transpose operation.
We call Fourier basis  the orthonormal basis obtained  by applying the discrete Fourier $F_n$  transform (DFT) to the computational basis of $\mathbb{C}^n$. The Fourier basis vectors  $ \ket{\tilde{j}}$ are the vectors of computational basis transformed by DFT, i.e.
\begin{equation}
    \ket{\tilde{j}}:=F_n^+\ket{j}=\frac{1}{\sqrt{n}}\sum\limits_{i=0}^{n-1} \omega^{-ij}\ket{i},\, j\in\{0,1,\ldots, n-1\}
\end{equation}
Given that $F_n\ket{\tilde{j}}=\ket{j},$ the Fourier basis diagonalizes the shift operator as $X_n\ket{\tilde{j}}=\omega^j \ket{\tilde{j}}$.
\\Now, using the basic notions from QIT introduced previously,  we focus on  developing the operator approach for generalized Cardano polynomials.

\begin{proposition}
    Let $n\geq3$ be an integer and  $c, d$ reals parameters( as well as $p$ and $q$ as given by \eqref{defp-q}). The Fujii operator $W$, as defined by
\begin{equation}\label{W}
    W:=pZ_n+qZ_n^{-1},
\end{equation}
has the following properties:
\begin{itemize}
    \item[i)] the standard computational basis vector $\ket{j}$ is an eigenvector of $W$ with eigenvalue $\lambda_j=p\omega^j+q\omega_j^{-1},$ for each $ j\in\{0,1,\ldots, n-1\}$. 
\item[ii)]  the operator  $W$ is normal and, for any polynomal $p\in\mathbb{C}(x),$ it holds that $p(W)=\sum\limits_{j}p(\lambda_j)\ket{j}\bra{j}$  
%(so, $p(w)$ is diagonal in the same basis as $W,$ with eigenvalues given by $p(\lambda_j)$. )
\item[iii)] It holds the operator identity $C_{n,c,d}(W)=0,$ where $C_{n,c,d}(x)$ is  the  n-th order generalized Cardano polynomial.
\end{itemize}
    \end{proposition}
\begin{proof}
 i)   The operator is diagonal by definition, as $Z_n=\text{ diag }[1,\omega, \omega^2,\ldots \omega^n],\, \forall n\geq 3$. Moreover, $Z_n$ is unitary and $Z_n^{-1}=Z_n^+,$ so $Z_n^{-1}\ket{j}=\omega^{-j}\ket{j}$. Consequently, the operator is also diagonal, as linear combination of two diagonal operators, that is $W=pZ_n+qZ_n^{-1}=\text{ diag }[p\omega^j+q\omega^{-j}], \, j\in\{0,1,\ldots, n-1\}.$ Therefore,
 the operator \eqref{W} has eigenvalues $\lambda_j=p\omega^j+q\omega^{-j}$.  Indeed, 
\begin{eqnarray}
    W\ket{j}=(pZ_n+qZ_n^{-1})\ket{j}=(p\omega^j+q\omega^{-j})\ket{j}.
\end{eqnarray} 
ii) The operator $W$ is a normal operator ($WW^+=W^+W$.) Given a polynom $p\in \mathbb{C}(T),$ it holds that $p(W)=\sum_j p(\lambda_j)\ket{j}\bra{j},$ so $p(W)$ is diagonal in the same basis and the eigenvalues are $p(\lambda_j)$. 
\\iii) Given that $\lambda_j=p\omega^j+q\omega^{-j},\,\forall j,$ are the roots $x[j]$ of $C_{n,c,d}(x),$ it follows that $C_{n,c,d}(W)=0, $ where the operator identity is defined  by functional calculus on diagonal matrices.
\end{proof}
The operator polynomial $C_{n,c,d}(W)=0$, that is 
\begin{equation}\label{CO}
    W^{2m+1}=2dI+\sum\limits_{j=0}^{m-1}B_{m,j}c^{m-j}W^{2j+1},
\end{equation}
is the generalization of Fujii's cubic and quartic operators\cite{Fujii2013} and its roots are $x[j],\in j\in\{0, \pm 1,\ldots, \pm m\}$. We notice that the operator $W$ encodes the roots structure of generalized Cardano polynomials $C_{n,c,d}(x)$. We will use the operator $W$ to define another operator, which is defined by Fourier conjugation of $W$.

\subsection{Cardano operator.}
Given the operator $W,$ associated to $n, c$ and $d$ as in \eqref{W}, we define the operator 
\begin{equation} \label{Cardano-operator}X:=F_n^+WF_n,\, n\geq 3\end{equation} and called it Cardano operator. It holds that
\begin{proposition}\label{th-X}
  The Cardano operator $X$ has the followings properties:  
  \begin{itemize}
      \item[a)] It can be written as $X=pX_n+qX_n^{-1}$ ( circular form),  where $X_n$ is the shift operator.
      \item[b)] The operator $X$ has the same spectrum as $W$.
      \item[c)] The operator $X$ satisfies the operator equation $C_{n,c,d}(X)=0$, where $C_{n,c,d}(X)$ is the n-th order Cardano- polynomial for  given real parameters $c$ and $d$, as given by eq.\eqref{gen-eq}.
      \end{itemize}
\end{proposition}
\begin{proof}
    a) Given that $F_n^+Z_nF_n=X_n$ and $F_n^+Z_n^{-1}F_n=X_n^{-1},$ it follows the the operator $X $ can be written as $X=pX_n+qX_n^{-1}$. Given the structure of the operators $X_n$ and $X_n^{-1},$ for which every row is a cycle shift of the first, we conclude that $X$ is a circulant matrix of the 
   form  $X=\circM(0,q,0,\ldots, 0,p)$.
    %\begin{center}
     %   $X=\begin{pmatrix}
      %  0&q&0&\ldots &0&p\\
       % p&0&q&\ldots&0&0\\
        %\vdots&\vdots&\vdots &\ldots &\vdots &\vdots\\
        %0&0&\ldots &p&0&q\\
        %q&0&\ldots &0&p&0
    %\end{pmatrix}$
    %\end{center}
\\b) Given its circular structure as well as the fact that is unitarily equivalent to $W$, the operator $X$ is diagonalizable in the Fourier eigenbasis and its eigenvalues are $\lambda_{\tilde{j}}=p\omega^{j}+q\omega^{-j}.$
\\c) As $F_n$ is a unitary matrix, it holds that $X^k=F_n^+W^kF_n=F_n^+ \text {diag }[\lambda^k[j]]F_n$. \\ We have that 
$C_n(X)=F_n^+C_n(W)F_n=0$.
\end{proof}
Let's explain now in details Fujii's construction\cite{Fujii2013}, that corresponds for the case $n=3$. Suppose that for given real parameters $c$ and $d$( and consequently $p$ and $q$), we construct the third order Cardano polynomial $C_{3,c,d}(x)=x^3-3cx-2d, $ whose roots are denoted by $x[j],\,j\in\{0,\pm1\}$.
Now\, relying on the operator approach, we introduced the operator $Z_3=\text {diag }[1, \omega, \omega^2]$ and   $W=pZ_3+qZ_3^{-1}=\text {diag }[p+q, p \omega +q \omega^{-1}, p\omega^2+q\omega^{-2}]=\text {diag }[\lambda_j],$ where $\lambda_j,\, j\in\{0,1,2\}$ denotes the eigenvalues of $W$. Suppose  that  $W=\text{ diag} [x[0],x[1],x[-1]],$ where $x[j]\, j\in\{0,\pm1\}$  are the roots of $C_{3,c,d}(x),$ as given by \eqref{j-root}. Using the discrete quantum Fourier transform, $F_3$,
the operator $X:=F_3^+WF_3$ has the following form \begin{equation}X=\frac{1}{3}\begin{pmatrix}
x[0]+x[1]+x[-1]&x[0]+x[1]\omega+x[2]\omega^2&x[0]+x[1]\omega^2+x[-1]\omega\\x[0]+x[1]\omega^2+x[-1] \omega &x[0]+x[1]+x[-1]&x[0]+x[1] \omega+x[-1] \omega^2\\
    x[0]+x[1] \omega+ x[-1] \omega^2&x[0]+x[1] \omega^2+x[-1]\omega& x[0]+x[1]+x[-1]
\end{pmatrix}\label{Fj}\end{equation} 

In the same time, $X=pX_3+qX_3^{-1}=\begin{pmatrix}
    0&q&p\\p&0&q\\q&p&0
\end{pmatrix}=\circM(0,q,p),$
where $X_3$ is generalized Pauli matrix. Consequently,  from the equality of the two forms of $X,$ it follows that

\begin{equation}
    \begin{pmatrix}
        0&q&p
    \end{pmatrix}=\begin{pmatrix}
       x[0]&x[1]&x[-1]
    \end{pmatrix}\frac{1}{3}\begin{pmatrix}
        1&1&1\\1&\omega&\omega^2\\1&\omega^2&\omega
    \end{pmatrix} \Leftrightarrow \begin{pmatrix}
        0&q&p
    \end{pmatrix}=\begin{pmatrix}
       x[0]&x[1]&x[-1]
    \end{pmatrix}\frac{1}{\sqrt{3}}F_3
\end{equation}
Therefore, the vector of the roots $x[j],\, j \in\{0,1,2\}$ is given by
\begin{equation}
    \begin{pmatrix}
        x[0]&x[1]&x[-1]
    \end{pmatrix}=\begin{pmatrix}
        0&q&p
    \end{pmatrix}\cdot \sqrt{3}F_3^{+}=\begin{pmatrix}
    p+q&p\omega+q\omega^2&p\omega^2+q\omega
    \end{pmatrix}.
\end{equation}
Consequently, we obtain the set of roots for the cubic Cardano  polynomial, as given by \eqref{j-root}, for $j\in\{0,\pm1\}$ and that they are equal to the eigenvalues of the operator $W,$ as well as of X (given that the two matrices are  similar).
Indeed, the above results follows naturally using \eqref{th-X}.  The Cardano operator cubic polynomial equation is $X^3-3cX-2dI=0 $ and  the eigenvalues of the Cardano operator  $X$, as given by \eqref{Cardano-operator}, are the roots of the polynomial $C_{3,c,d}(x).$ 
\par The above construction can be extended for any $n$ odd natural number.  The Cardano operator $X=F_n^+WF_n$ has the eigenvalues given by the roots \eqref{j-root} of the Cardano polynomial $C_{n,c,d}(x),$ associated to the polynomial equation \eqref{gen-eq1}.
\iffalse
\subsection{Ferrari operator.}
In this section we show that the operator‑theoretic Ferrari construction for quartic equations is a special case of the general spectral method developed earlier for Cardano polynomials. The following procedure reconsider in the operator setting the procedure presented in section\eqref{Ferrai-eq}, for solving quartic Ferrari equations. Indeed, Ferrari’s method reduces the quartic operator equations
\begin{equation}
    X^4+A_2X^2+A_1X+A_0I=0,
\end{equation}
to the cubic Cardano equation $Z^3-3CZ-2DI=0,$ where $C$ and $D$ are functions of initial real coefficients $A_i,\, i\in\{0,1,2\}$.
Using the Cardano operator $W=pZ_3+qZ_3^{-1},$ where $p$ and q are associated to the parameters $C$ and $D,$ we can solve the cubic depressed Cardano equation and therefore find the solution of Ferrari operator equation using  a similar procedure as presented in section \eqref{Ferrai-eq}, by transforming  the scalar equations into operator identities.
\vspace{0.3cm}
\fi

\section{Cardano-generated POVMs}
\label{sec:CardanoPOVMs}
\subsection{ Construction of Cardano-generated POVMs.}
We show now that the theory of generalized Cardano polynomials is related to  important concept of QIT, such as  POVMs( positive-operator valued measurements)\cite{Heinosaari2016}. Indeed,  the spectral structure associated with the generalized
Cardano polynomial  can be exploited to introduce a generating procedure for constructing quantum
measurements, whose effects are positive normalized functions of the Cardano spectrum.  
\begin{proposition}
    Given the spectral decomposition of the generalized Cardano (Fujii)
operator $
W=\sum\limits_{j=0}^{n-1}\lambda_jP_j$,
 where
$
P_jP_k=\delta_{jk}P_j,
\,
\sum\limits_{j=0}^{n-1}P_j=I$ and a family of functions $
f_\alpha:\operatorname{Spec}(W)\longrightarrow [0,1],
\text{
satisfying }
\sum\limits_{\alpha}f_\alpha(\lambda_j)=1,
\, j=0,\ldots,n-1,$ then the operators $
E_\alpha:=f_\alpha(W)$  form a POVM.
\end{proposition}
\begin{proof}
    Using the fact that the operator $W$ satisfies the identity $
C_{n,c,d}(W)=0
,$ and, by  spectral functional calculus, as well as  by the fact that $ \sum\limits_{\alpha}f_\alpha(\lambda_j)=1,
\, j=0,\ldots,n-1,$  it follows that 
$ E_\alpha
=
\sum\limits_{j=0}^{n-1}f_\alpha(\lambda_j)P_j,
$
and hence \(E_\alpha\geq0\) and
$
\sum_\alpha E_\alpha=I.$
Consequently, $
\mathsf E=\{E_\alpha\}_\alpha $
is a POVM.
\end{proof}

The construction is particularly natural from the algebraic point of
view. Since $
C_{n,c,d}(W)=0, $
every polynomial in \(W\) can be reduced modulo \(C_{n,c,d}\) to a
polynomial of degree at most \(n-1\). More explicitly, if
\[
f_\alpha(x)
=
q_\alpha(x)C_{n,c,d}(x)+r_\alpha(x),
\qquad
\deg r_\alpha\leq n-1,
\]
then
\[
E_\alpha=f_\alpha(W)=r_\alpha(W).
\]
Writing $
r_\alpha(x)
=
\sum_{k=0}^{n-1}a_{\alpha,k}x^k,$
we obtain
\[
E_\alpha
=
\sum_{k=0}^{n-1}a_{\alpha,k}W^k.
\]
The coefficients \(a_{\alpha,k}\) are therefore the coefficients of
the remainder of the chosen measurement function \(f_\alpha\) modulo
the generalized Cardano polynomial. 

When the eigenvalues \(\lambda_0,\ldots,\lambda_{n-1}\) are distinct,
the spectral projections can furthermore be expressed by Lagrange
interpolation,
\[
P_j
=
\prod_{\substack{k=0\\ k\neq j}}^{n-1}
\frac{W-\lambda_k I}{\lambda_j-\lambda_k}.
\]
Consequently, the effects admit the explicit representation
\[
E_\alpha
=
\sum_{j=0}^{n-1}
f_\alpha(\lambda_j)
\prod_{\substack{k=0\\ k\neq j}}^{n-1}
\frac{W-\lambda_k I}{\lambda_j-\lambda_k}.
\]
For a quantum state \(\rho\), the probability of the outcome
\(\alpha\) is
\[
\Pr(\alpha|\rho)
=
\operatorname{Tr}(\rho E_\alpha)
=
\sum_{j=0}^{n-1}
f_\alpha(\lambda_j)
\operatorname{Tr}(\rho P_j).
\]
Hence a Cardano-generated POVM can be interpreted as a classical
post-processing of the spectral measurement of \(W\). If
\(t_{\alpha j}=f_\alpha(\lambda_j)\), then
$
t_{\alpha j}\geq0,
\,
\sum_\alpha t_{\alpha j}=1,
$
so that \(T=(t_{\alpha j})\) is a stochastic matrix and
\[
E_\alpha=\sum_{j=0}^{n-1}t_{\alpha j}P_j.
\]
This provides a direct link between the generalized Cardano operator
algebra and the operational description of quantum measurements.

In the trigonometric regime considered above, where
$
\lambda_j
=
2\sqrt{c}\,
\cos\left(\frac{\alpha+2\pi j}{n}\right),
$
the functions \(f_\alpha\) are defined on this explicitly known finite
set. Thus the parameters of the generalized Cardano polynomial
determine the spectral points on which the POVM construction is based,
while the choice of the functions \(f_\alpha\) determines the
measurement statistics. In this sense, the identity
\(C_{n,c,d}(W)=0\) does not merely encode the roots of the generalized
Cardano polynomial; it also provides a finite operator algebra from
which families of quantum measurements can be constructed.

\subsection{Example: a binary POVM from the cubic Cardano operator}

We illustrate the construction with the simplest nontrivial case, namely
$n=3$, where the generalized Cardano polynomial is given by the polynomial $
C_{3,c,d}(x)=x^3-3cx-2d.$

The trigonometric parametrization
$
\lambda_j
=
2\sqrt{c}\cos\left(\frac{\alpha+2\pi j}{3}\right),
\, j=0,1,2,$
for 
$
c=1,\, \alpha=\frac{\pi}{2},$
yields
\[
\lambda_0=\sqrt{3},\qquad
\lambda_1=-\sqrt{3},\qquad
\lambda_2=0.
\]
Since $d=\cos\alpha=0$, the corresponding Cardano polynomial is
$
C_{3,1,0}(x)=x^3-3x, $
and hence the generalized Cardano operator $W$ satisfies
$
W^3-3W=0.
$

Consider  the spectral decomposition $
W=\sum_{j=0}^{2}\lambda_jP_j
=
\sqrt{3}\,P_0-\sqrt{3}\,P_1,$
where $P_0,P_1,P_2$ are the spectral projections associated with
$\sqrt{3},-\sqrt{3},0$, respectively, and  $
P_jP_k=\delta_{jk}P_j,
\,
\sum\limits_{i=0}^2P_i=I$.

We now choose two functions on the spectrum of $W$,
\[
f_+(x)=\frac{x+\sqrt{3}}{2\sqrt{3}},
\qquad
f_-(x)=1-f_+(x)
       =\frac{\sqrt{3}-x}{2\sqrt{3}}.
\]
Their values on the spectrum are
\[
\begin{array}{c|ccc}
x & \sqrt{3} & 0 & -\sqrt{3}\\
\hline
f_+(x) & 1 & \frac12 & 0\\
f_-(x) & 0 & \frac12 & 1
\end{array}.
\]
Consequently, the corresponding effects are
\[
E_+=f_+(W)
=
\frac{W+\sqrt{3}\,I}{2\sqrt{3}},
\qquad
E_-=f_-(W)
=
\frac{\sqrt{3}\,I-W}{2\sqrt{3}}.
\]
Using the spectral decomposition of $W$, these operators can also be
written as
\[
E_+
=
P_0+\frac12P_2,
\qquad
E_-
=
\frac12P_2+P_1.
\]
It is immediate to see that the operators $E_{+}$ and $E_{-}$ satisfy
$
E_+\geq0,\qquad E_-\geq0,\,
E_++E_-=I.$
Therefore, $
\{\{E_+,E_-\}$
is a binary POVM generated by the generalized Cardano operator $W$.

For a quantum state $\rho$, the two outcome probabilities are
\[
\Pr(+\mid\rho)
=
\operatorname{Tr}(\rho E_+)
=
\operatorname{Tr}(\rho P_0)
+
\frac12\operatorname{Tr}(\rho P_2),
\]
and
\[
\Pr(-\mid\rho)
=
\operatorname{Tr}(\rho E_-)
=
\frac12\operatorname{Tr}(\rho P_2)
+
\operatorname{Tr}(\rho P_1).
\]
Thus the POVM can be interpreted as a classical post-processing of the
spectral measurement $\{P_0,P_1,P_2\}$.  The eigenstates corresponding
to $\pm \sqrt{3}$  give deterministic outcomes, whereas the
zero-eigenspace produces the two outcomes with equal probability.

This example also illustrates the polynomial realization of Cardano-generated
POVMs.  Since
\[
C_{3,1,0}(W)=0,
\qquad\text{equivalently}\qquad
W^3=3W,
\]
every polynomial in $W$ can be reduced to degree at most two.  In the
present case the chosen measurement functions are already affine, and
hence
\[
E_+
=
\frac12I+\frac{1}{2\sqrt{3}}W,
\qquad
E_-
=
\frac12I-\frac{1}{2\sqrt{3}}W.
\]
The construction therefore exhibits explicitly the chain
\[
C_{3,1,0}(W)=0
\quad\Longrightarrow\quad
\operatorname{Spec}(W)=\{\pm\sqrt{3},0\}
\quad\Longrightarrow\quad
\{E_+,E_-\}\ \text{is a POVM}.
\]
This provides a concrete quantum-information-theoretic realization of
the finite operator algebra generated by the generalized Cardano
polynomial.
\section{Conclusions}\label{s:cl}
This paper introduces a generalization of Cardano's formula for a subclass of higher-order polynomials. Although the technique can be used to solve all polynomials of order 3, only a subset of polynomials of higher order can be solve. We construct this family of two-parameters odd-order polynomials and called them generalized Cardano polynomials. We also use our method to 
 extend Fujii's operator formulation\cite{Fujii2013} of Cardano method beyond cubic case. Our approach  show that the algebraic structure of two-parameters odd order polynomials can be expressed entirely within the algebra generated by the clock operator $Z_n$.
The resulting operator $W$ satisfies the general odd-degree Cardano polynomial identity \eqref{CO} and it provides in a single, closed form a family operator polynomials, that are connected to a family of two-parameters odd-degree equations.
We unify the classical algebraic methods to that of operator formulation which, based on spectral analysis, offers similar results.
The operator Cardano approach offers a bridge between classical algebra and quantum operator theory and a practical way to encode this particular polynomial structure into the finite-dimensional operators.  We show that, based on Cardano operator, we can find  an mechanism to solve a class of two-parameters odd‑degree polynomials equations and highlights structural connections between algebraic solvability, circulant operators, and quantum‑inspired operator calculus.  Moreover, Cardano-generated POVMs are constructed using spectral properties of Cardano operator and an example of a binary POVM for cubic Cardano operator is presented.
We believe that this work may be explored as  applications of operators identities in quantum computations/information procedures, such as quantum algorithms, spectral engineering, where polynomials relations, Fourier diagonalization and POVMs play central role.
\vspace{0.3cm}
\\
\noindent\textit{Acknowledgements.} 
The work of MAJ was supported by RExQTCS: Romanian Excellence in Quantum Technologies enhancing Cybersecurity, a grant of the Ministry of Education and Research from Romania, CCCDI - UEFISCDI, project number PN-IV-P6-6.1-CoEx-2024-0214, within PNCDI IV. 
\\\textit{AI Usage.} Several LLMs have been used interactively on author-defined models, assumptions, equations, and tasks,
with all outputs independently checked and corrected where
necessary by the authors.
\section{Appendix}
\label{s:A}
In the following we are going to give a complete proof for the polynomial equation that characterizes the  Cardano polynomial $C_{n,c,d}(x)$. The binomial formula applied to $x=p+q,$ as given by \eqref{prop-p-q} is
    \begin{eqnarray}\label{bin}
    x^{n}=(p+q)^n=p^n+q^n+\sum\limits_{k=1}^{m }\begin{pmatrix}2m+1\\k\end{pmatrix}(pq)^{k}(p^{2(m-k)+1}+q^{2(m-k)+1})\label{bin-formula}.\end{eqnarray}
    The term $p^r+q^r, \, \forall r\in\mathbb{N^*}$ can be written in terms of $x=p+q$ and $xy=c $ as given by Kummer formula\cite{Ma2005}
\begin{equation}\label{Krummer}
    p^r+q^r=\sum\limits_{i=0}^{[r/2]}(-1)^i\frac{r}{r-i}\begin{pmatrix}
        r-i\\i
    \end{pmatrix}(pq)^i(p+q)^{r-2i}=\sum\limits_{i=0}^{[r/2]}(-1)^i\frac{r}{r-i}\begin{pmatrix}
        r-i\\i
    \end{pmatrix}c^ix^{r-2i},
\end{equation}
Consequently, eq. \eqref{bin} becomes
\begin{eqnarray}
x^n=2d+\sum\limits_{k=1}^{m}\sum\limits_{i=0}^{m-k} (-1)^i\begin{pmatrix}2m+1\\k\end{pmatrix}\frac{2(m-k)+1}{2(m-k)+1-i}\begin{pmatrix}
        2(m-k)+1-i\\i
    \end{pmatrix}c^{i+k}x^{2(m-k)+1 -2i}
    \label{gen-eq1}
\end{eqnarray}
By changing the index $i\mapsto j$, where $j=m-k-i,$ we recover  the following form
\begin{eqnarray}
x^n=2d+\sum\limits_{k=1}^{m}\sum\limits_{j=0}^{m-k} (-1)^{m-k-j}\begin{pmatrix}2m+1\\k\end{pmatrix}\frac{2(m-k)+1}{m-k+1+j}\begin{pmatrix}
        m-k+1+j\\m-j-k
    \end{pmatrix}c^{m-j}x^{1+2j}
    \label{gen-eq11}.
\end{eqnarray}

For fixed $j\in\{0,\ldots, m-1\}, $ we will denote the coefficient of $c^{m-j}x^{2j+1}$ as follows

\begin{eqnarray}
B_{m,j}=\sum\limits_{k=1}^{m-j}(-1)^{m-j-k}\begin{pmatrix}2m+1\\k\end{pmatrix}\frac{2(m-k)+1}{m-k+1+j}\begin{pmatrix}
        m-k+1+j\\1+2j
    \end{pmatrix}
    \label{B-term1}.
\end{eqnarray}
 Let's denote $t=m-j-k$. Then,

\begin{eqnarray}
B_{m,j}=\frac{1}{2j+1}\sum\limits_{t=0}^{m-j-1}(-1)^{t}\begin{pmatrix}2m+1\\m-j-t\end{pmatrix}[2(j+t)+1]\begin{pmatrix}
        2j+t\\t
    \end{pmatrix}
    \label{B-term12}.
\end{eqnarray}
Let's complete the summation in $B_{m,j}$ up to $(m-j)$ and denote it by $S_{m,j}$. Then,
\begin{equation}\label{b}
B_{m,j}=S_{m,j}-(-1)^{m-j}\frac{2m+1}{2j+1}\begin{pmatrix}
    m+j\\m-j
\end{pmatrix}.\end{equation}
We are going to show the the term $S_{m,j}=\frac{1}{2j+1}\sum\limits_{t=0}^{m-j}(-1)^{t}\begin{pmatrix}2m+1\\m-j-t\end{pmatrix}[2(j+t)+1]\begin{pmatrix}
        2j+t\\t
    \end{pmatrix}$ reduces to 0. To this aim, we notice that the term $(-1)^t(2j+2t+1)\begin{pmatrix}
        2j+t\\t
    \end{pmatrix}$ can be found by using the expansion $\sum\limits_{t=0}^\infty \begin{pmatrix}
        2j+t\\t
    \end{pmatrix}(-x)^t=(1+x)^{-(2j+1)}, 
    \, |x|<1$ by applying the differential operator $2x\frac{d}{dx}+(2j+1).$
    Indeed, by applying the the differential operator on the left-hand side of the expansion we get 
\begin{eqnarray}\nonumber
     [2x\frac{d}{dx}+(2j+1)]\left(\sum\limits_{t=0}^{\infty}\begin{pmatrix}
         2j+t\\t
\end{pmatrix} (-x)^t\right)=\sum\limits_{t=0}^{\infty}(2j+2t+1)\begin{pmatrix}
        2j+t\\t
    \end{pmatrix}(-x)^t,
\end{eqnarray}

whereas, applying on the right-hand side of the expansion, it yields \begin{equation}[2x\frac{d}{dx}+(2j+1)] \left((1+x)^{-(2j+1)}\right)=\frac{(2j+1)(1-x)}{(1+x)^{2j+2}}\nonumber .\end{equation}
Therefore, the sum $S_{m,j}$ is the coefficient of $x^n, \,( n=m-j) $ in the product of the above generating function and the generating function for the remaining binomial coefficient of $(1+x)^{2m+1}$. Here $[x^n][P(x)]$ denotes the coefficient of $x^n$ from the polynomial $P(x)$. 
Hence, we have
\begin{eqnarray}
    S_{m,j}=\frac{1}{2j+1}[x^n]\cdot \left[(1+x)^{2m+1}\frac{(2j+1)(1-x)}{(1+x)^{2j+2}}\right]\nonumber.
\end{eqnarray}
  Consequently, 
 \begin{eqnarray}
     S_{m,j}=
   [x^n]\left[(1-x)(1+x)^{2n-1}\right]=  [x^n]\left[(1+x)^{2n-1}\right]-[x^{n-1}]\left[(1+x)^{2n-1}\right], \nonumber
 \end{eqnarray}
which are the following binomial coefficients
$S_{m,j}=\begin{pmatrix}
         2n-1\\n
     \end{pmatrix}-\begin{pmatrix}
         2n-1\\n-1
     \end{pmatrix}=0$.
 Therefore, eq.\eqref{b} is going to be equal to \eqref{B}.
% BIBLIOGRAPHY: Create or upload a bib file to manage your references. For instructions, see www.overleaf.com/learn/how-to/Using_bibliographies_on_Overleaf. Citations will use the Chicago Manual of Style notes and bibliography citation format, described at www.chicagomanualofstyle.org/tools_citationguide/citation-guide-1.html.
%\bibliographystyle{alpha}
%\bibliography{tensors}
%\printbibliography

\end{document}